\documentclass[12pt]{article}

\usepackage{amsmath}
\usepackage{amssymb}
\usepackage{amsfonts}
\usepackage{latexsym}
\usepackage{color}

\catcode `\@=11 \@addtoreset{equation}{section}

\catcode `\@=12

\newcommand{\be}{\begin{equation}}
\newcommand{\en}{\end{equation}}
\newcommand{\ee}{\end{equation}}
\newcommand{\bea}{\begin{eqnarray}}
\newcommand{\ena}{\end{eqnarray}}
\newcommand{\beano}{\begin{eqnarray*}}
\newcommand{\enano}{\end{eqnarray*}}
\newcommand{\bee}{\begin{enumerate}}
\newcommand{\ene}{\end{enumerate}}

\newcommand{\M}{\mathfrak M}

\newcommand{\mc}{\mathcal}

\newcommand{\E}{{\cal E}}
\newcommand{\F}{{\cal F}}

\newcommand{\1}{1 \!\! 1}

\newcommand{\Hil}{\mc H}

\newtheorem{thm}{Theorem}

\newtheorem{defn}{Definition}

\newenvironment{proof}{\noindent {\bf Proof. }}{\hfill$\square$ \vspace{3mm}\endtrivlist}

\catcode `\@=11 \@addtoreset{equation}{section}
\catcode `\@=12

\newcommand{\kt}{\rangle}
\newcommand{\br}{\langle}

\def\ben{$$}
\def\een{$$}
\def\ba{\begin{array}{c}}
\def\ea{\end{array}}

\begin{document}

\thispagestyle{empty}

\vspace*{1cm}

\begin{center}
{\Large \bf  Non linear pseudo-bosons versus hidden Hermiticity}\\[10mm]

{\large Fabio Bagarello}\\
  Dipartimento di Metodi e Modelli Matematici,
Facolt\`a di Ingegneria,\\ Universit\`a di Palermo, I-90128  Palermo, Italy\\
e-mail: bagarell@unipa.it\\
home page:
www.unipa.it/$\sim$bagarell\\

\vspace{3mm}

and

 \vspace{3mm}

{\large Miloslav Znojil}\\
  Nuclear Physics Institute ASCR,
 250 68 \v{R}e\v{z}, Czech Republic\\
e-mail: znojil@ujf.cas.cz\\
home page:
http://gemma.ujf.cas.cz/$\sim$znojil\\

\vspace{3mm}

\end{center}

\vspace*{2cm}

\begin{abstract}
 \noindent
The increasingly popular concept of a ``hidden" Hermiticity of
operators is compared with the recently introduced notion of {\em
non-linear pseudo-bosons}. The formal equivalence between these two
notions is deduced under very general
assumptions. Examples of their applicability in quantum
mechanics are discussed.

\end{abstract}

\vspace{2cm}


\vfill


\newpage

\section{Introduction}

The strong formal limitations imposed upon an {\em observable} operator $U$ by the
requirement of its Hermiticity in a suitable Hilbert space ${\cal
H}$,  $U = U^\dagger$, have long been perceived as a challenge. In
mathematics, for example, Dieudonn\'{e} \cite{Dieudonne} introduced
the notion of the so called quasi-Hermiticity of $U$ based on the
weakened requirement $ T\,U = U^\dagger T$ with a suitable $T> 0$.
He emphasized that without additional assumptions, unexpectedly,
 the
adjoint $U^\dagger$ is not quasi-Hermitian in general, so
that the spectrum of $U$ is not necessarily real.

In the context of physics (and, most typically, in quantum
mechanics) one usually tries to avoid similar paradoxes by accepting
additional assumptions. In this context, one of the most successful
attempted generalizations of the Hermiticity has been proposed by
Scholtz et al \cite{Geyer} who restricted their attention only to
bounded observables  $U \to A \in {\cal B}({\cal H})$ and to bounded
``metric" operators $T, T^{-1} \in {\cal B}({\cal H})$.

Unfortunately, being inspired and guided just by the well known
terminology used in linear algebra of the $N$ by $N$ matrices $A$
acting in the finite-dimensional Hilbert space ${\cal H}^{(N)}
\equiv \mathbb{C}^N$, the authors of Ref.~\cite{Geyer} gave their
very narrow compact-operator subset of the Dieudonn\'{e}'s
quasi-Hermitian set {\em the same} name. In our recent summary and
completion of their proposal \cite{SIGMA}, therefore, we slightly
modified the notation (replacing the symbol for ``metric", $T \to
\Theta$) and, having recalled the Smilga's innovative terminology
\cite{Smilga}, we recommended to call the corresponding and very
nicely behaved quasi-Hermitian operators $A$ {\em cryptohermitian}.

Originally, the concept of cryptohermiticity (meaning, in essence,
just a hidden form of Hermiticity \cite{ali}) proved successful just
in the area of physics of heavy nuclei \cite{Geyer}. About fourteen
years ago a ``new life" of this concept has been initiated by Bender
with coauthors. In a way inspired by the needs of quantum field
theory \cite{DB,BM} and along the path independent of
Ref.~\cite{Geyer} they proposed an innovation of textbooks on
quantum theory. Under the nickname of  ${\cal PT}-$symmetric quantum
theory  their formalism may be found described, e.g., in
reviews \cite{ali,Dorey,Carl}.

Briefly, this formalism may be characterized by the
heuristically fruitful postulate of the so called ${\cal
PT}-$symmetry of observables (here, ${\cal P}$ means parity while
${\cal T}$ denotes the antilinear operator of time reversal
\cite{BB}). Secondly, the formalism replaces the most common
physical assumption of the reality (i.e., observability) of the
argument $x$ of the wave function $\psi(x)$ by the observability of
the so called charge ${\cal C}$. Ultimately, ${\cal PT}-$symmetric quantum
theory eliminates the
well known interpretation ambiguities of the general cryptohermitian
quantum theory \cite{Geyer,SIGMAdva} by the recommended selection of
the physical metric in the unique, factorized, ${\cal
CPT}-$symmetry-mimicking form of product $\Theta^{({\cal
CPT})}={\cal PC}$.

It is worth mentioning that the ideas coming from ${\cal
PT}-$symmetric quantum theory proved inspiring and influential even
far beyond quantum physics \cite{Makris,MHD,eldyn}. At the same
time, the modified and extended forms of the cryptohermiticity with
$\Theta \neq {\cal PC}$ had to be used for the description of the
manifestly time-dependent quantum systems \cite{timedep} and/or in
the context of the scattering dynamical regime \cite{Jones,fund}. In
our present paper we intend to pay attention to another alternative
to the formalism represented by the recent independent and parallel
introduction and studies of the concept of the so-called {\em
pseudo-bosons} (PB).

The most compact presentation of the latter PB concept is due to
Trifonov \cite{tri}. A deeper understanding of its mathematics has
been provided by the very recent series of  papers
\cite{bagpb1,bagpb2,bagpb3,bagcal,bagpbJPA,bagpb4,abg,bagrev} where
one of us (FB) revealed the necessity as well as the key importance
of the fully rigorous treatment of some of the underlying formal
questions. We intend to continue these studies in what follows.

The open and interesting formal questions arise from the bosonic
form of the canonical commutation relation $[a,a^\dagger]=\1$ upon
replacing $a^\dagger$ by another (unbounded) operator $b$ not (in
general) related to $a$: $[a,b]=\1$. More recently, FB has extended
the general settings to what has been called {\em non linear
pseudo-bosons} (NLPB, \cite{bagnlpb}) where the role of the
commutation relation is replaced by a different requirement (see
below). This extension is motivated by the attempt to include, in
this general settings, hamiltonian-like operators which have a {\em
rich} spectrum, and in particular eigenvalues, labeled by a set of
quantum numbers $n_1$, $n_2$, $\ldots$, which are not linear
functions of $n_j$'s. An interesting aspect of this construction is
the possibility of getting operators ($M$ and $M^\dagger$) which are
not self-adjoint but still have real eigenvalues. This peculiarity
is well explained by the presence of an intertwining operator (IO)
between, say, $M$ and a third self-adjoint operator, $\tilde M$.
General results on IO show that, in this case, $M$ and $\tilde M$
are isospectral, and their eigenvectors are also related by the IO
itself.

On the other hand,  the use of the above-mentioned notion of
cryptohermiticity (CH) of a given operator opens the possibility of
the parallel work with a given operator using its parallel
representations in several Hilbert spaces, mutually not necessarily
related by a unitary transformation. In this setting it is obvious
that the PB and CH concepts may be related, meaning that the
pseudo-bosonic settings provide examples of general statement
introduced for cryptohermitian operators, or vice versa.

In our present paper we intend to proceed with this analysis,
showing a sort of equivalence between NLPB and CH. In particular, we
will show that, under very reasonable assumptions, any
cryptohermitian operator gives rise to a family of NLPB which are
{\em regular} (NLRPB), see below, and vice-versa, each family of
NLRPB produces in a natural way a cryptohermitian operator.

The paper is organized as follows:  in the next section, after a
short introduction to NLRPB, we prove the  equivalence outlined
above between these excitations and a certain cryptohermitian
operator. Sections III and IV are devoted to examples, while our
conclusions are given in Section V.


\section{NLRPB versus cryptohermiticity}

We begin this Section with a short review of NLPB, giving some
details in particular on the role of bounded or unbounded operators.
We refer to \cite{bagnlpb} for some preliminary examples of this
construction. Other examples will be discussed in Sections III and
IV.

\subsection{Non linear RPB}

In \cite{bagnlpb} FB has used the main ideas which produce, out of
coherent states, the so-called non-linear coherent states, to extend
the original framework proposed for pseudo-bosons to what he has
called {\em non-linear pseudo-bosons}. The starting point is a
strictly increasing sequence $\{\epsilon_n\}$  such that
$\epsilon_0=0$: $0=\epsilon_0<\epsilon_1<\cdots<\epsilon_n<\cdots$.
Then, given two operators $a$ and $b$ on the Hilbert space $\Hil$,

\begin{defn}
We will say that the triple $(a,b,\{\epsilon_n\})$ is a family of NLRPB if the following properties hold:
\begin{itemize}

\item {\bf p1.} a non zero vector $\Phi_0$ exists in $\Hil$ such that $a\,\Phi_0=0$ and $\Phi_0\in D^\infty(b)$.

\item {\bf { p2}.} a non zero vector $\eta_0$ exists in $\Hil$ such that $b^\dagger\,\eta_0=0$ and $\eta_0\in D^\infty(a^\dagger)$.

\item {\bf { p3}.} Calling
\be
\Phi_n:=\frac{1}{\sqrt{\epsilon_n!}}\,b^n\,\Phi_0,\qquad \eta_n:=\frac{1}{\sqrt{\epsilon_n!}}\,{a^\dagger}^n\,\eta_0,
\label{55}
\en
we have, for all $n\geq0$,
\be
a\,\Phi_n=\sqrt{\epsilon_n}\,\Phi_{n-1},\qquad b^\dagger\eta_n=\sqrt{\epsilon_n}\,\eta_{n-1}.
\label{56}\en
\item {\bf { p4}.} The sets $\F_\Phi=\{\Phi_n,\,n\geq0\}$ and $\F_\eta=\{\eta_n,\,n\geq0\}$ are bases of $\Hil$.

\item {\bf { p5}.} $\F_\Phi$ and $\F_\eta$ are Riesz bases of $\Hil$.

\end{itemize}

\end{defn}

As noticed in \cite{bagnlpb},  the definitions in (\ref{55}) are well posed in the sense that, because of {\bf p1} and {\bf p2}, the vectors $\Phi_n$ and $\eta_n$ are well defined vectors of $\Hil$ for all $n\geq0$. Moreover, but for {\bf p3}, the other conditions above coincide exactly with those of RPB. In fact, we can show that {\bf p3} replaces (and extends) the commutation rule $[a,b]=\1$, which is recovered if $\epsilon_n=n$. Moreover,  if all but {{\bf p5}} are  satisfied, then we have called our particles NLPB.

Let us introduce the following (not self-adjoint) operators:
\be
M=ba,\qquad \M=M^\dagger=a^\dagger b^\dagger.
\label{57}\en
Then we can check that $\Phi_n\in D(M)\cap D(b)$, $\eta_n\in D(\M)\cap D(a^\dagger)$, and, more than this, that
\be
b\,\Phi_n=\sqrt{\epsilon_{n+1} }\,\Phi_{n+1},\qquad a^\dagger\eta_n=\sqrt{\epsilon_{n+1} }\,\eta_{n+1},
\label{58}\en
which is a consequence of definitions (\ref{55}), as well as
\be
M\Phi_n=\epsilon_n\Phi_n,\qquad \M\eta_n=\epsilon_n\eta_n,
\label{59}\en
These eigenvalue equations imply that the vectors in $\F_\Phi$ and $\F_\eta$ are mutually orthogonal. More explicitly,
\be
\left<\Phi_n,\eta_m\right>=\delta_{n,m}.
\label{510}\en
where  we have fixed the normalization of $\Phi_0$ and $\eta_0$ in such a way that $\left<\Phi_0,\eta_0\right>=1$.

In \cite{bagnlpb} we have also proved that conditions  \{{\bf p1}, {\bf p2}, {\bf p3}, {\bf p4}\} are equivalent to  \{{\bf p1}, {\bf p2}, {\bf p3}$'$, {\bf p4}\}, where

\vspace{2mm}

{\bf {\bf p3}$'$.} The vectors $\Phi_n$ and $\eta_n$ defined in (\ref{55}) satisfy (\ref{510}).

\vspace{2mm}

In the following, therefore, we can use {\bf p3} or {\bf p3}$'$
depending of which is more convenient for us.

\vspace{2mm}

Carrying on our analysis on the consequences of the definition on
NLRPB, and in particular of ${\bf p4}$, we rewrite this assumption
in bra-ket formalism as
 \be
\sum_n|\Phi_n\kt\,\br \eta_n|=\sum_n|\eta_n\kt\,\br \Phi_n|=\1,
\label{512}
 \en
while {\bf p5} implies that the operators
$S_\Phi:=\sum_n|\Phi_n\kt\,\br \Phi_n|$ and
$S_\eta:=\sum_n|\eta_n\kt\,\br \eta_n|$ are positive, bounded,
invertible and that $S_\Phi=S_\eta^{-1}$.  The new fact is that the
operators $a$ and $b$ do not, in general, satisfy any {\em simple}
commutation rule. Indeed, we can check that, for all $n\geq0$,
 \be
[a,b]\Phi_n=\left(\epsilon_{n+1}-\epsilon_n\right)\Phi_n,
\label{513}\en
which is different from $[a,b]=\1$ in general. In \cite{bagnlpb} it has also been proved that, not surprisingly, if  $\sup_n\epsilon_n=\infty$, then the operators $a$ and $b$ are unbounded. We end this overview  mentioning also that $M$ and $\M$ are connected by an intertwining operator, related to $S_\Phi$. We will use this property in what follows.

\subsection{Connection with cryptohermiticity}

The introduction of the above-mentioned notion of cryptohermiticity
of a given operator $H$ \cite{SIGMA} enables us to distinguish
between the ``first" Hilbert space  $\Hil^{(F)}$ (in which the
operator in question is {\em not} self-adjoint, $H\neq H^\dagger$)
and the ``second" Hilbert space  $\Hil^{(S)}$ in which {\em the
same} operator  {\em is} self-adjoint (one may write, e.g.
\cite{SIGMA}, $H =H^\ddagger$). The idea behind such an apparent
paradox is that one can say that $H$ can only be declared
self-adjoint {\em with respect to a definite scalar product}. In
this sense one usually starts form the ``friendly" definition of the
so called Dirac's (i.e., roughly speaking, ``transposition plus
complex conjugation") definition of $H^\dagger$ in $\Hil^{(F)}$ and
complements its by the mere modification of the inner product in
$\Hil^{(F)}$ yielding the explicit definition of
$H^\ddagger=\Theta^{-1}H^\dagger\Theta$ written in terms of the
positive metric operator $\Theta=\Theta^{(S)} \neq I$ which should
be, together with its inverse \cite{Geyer}, bounded and self-adjoint
in the ``friendly" space $\Hil^{(F)}$.

In the models where $\Theta$ as well as $\Theta^{-1}$ are bounded,
one can comparatively easily deal with mathematical questions.
Otherwise, there emerge several subtle points related to the domain
of the operators in question.  Due to the relevance of the {metric
operator} $\Theta$, let us make now the standard notation
conventions less ambiguous.

\begin{defn}
Let us consider two (not necessarily bounded) operators $H$ and
$\Theta$ acting on the Hilbert space $\Hil$, with $\Theta$  positive
and invertible. Let us call $H^\dagger$ the adjoint of $H$ in $\Hil$
with respect to its scalar product and
$H^\ddagger=\Theta^{-1}H^\dagger\Theta$, when this exists. We will
say that $H$ is cryptohermitian with respect to $\Theta$
(CHwrt$\Theta$) if $H=H^\ddagger$.
\end{defn}

Using standard facts on functional calculus it is obvious that the operators $\Theta^{\pm 1/2}$ are well defined. Hence we can introduce an operator $h:=\Theta^{1/2}\,H\,\Theta^{-1/2}$, at least if the domains of the operators allow us to do so. More explicitly, $h$ is well defined if, taken $f\in D(\Theta^{-1/2})$,
$\Theta^{-1/2}f\in D(H)$ and if $H\,\Theta^{-1/2}f\in D(\Theta^{1/2})$. Of course,
these requirements are surely satisfied if $H$ and $\Theta^{\pm 1/2}$ are bounded. Otherwise some care is required. It is easy to check that $h=h^\dagger$. Hence the following definition appears natural:

\begin{defn}
Assume that  $H$ is CHwrt$\Theta$, for $H$ and $\Theta$ as above. $H$ is {\em well behaved} wrt $\Theta$ if $h$ has only discrete eigenvalues $\epsilon_n$, $n\in {\Bbb N}_0:={\Bbb N}\cup\{0\}$, with eigenvectors $e_n$: $he_n=\epsilon_n e_n$, $n\in {\Bbb N}_0$.
\end{defn}

\vspace{2mm}

It is convenient, but not really necessary, to restrict ourself to
the case in which the multiplicity of each eigenvalue $\epsilon_n$,
$m(\epsilon_n)$, is one. To fix the ideas we also assume that
$0=\epsilon_0<\epsilon_1<\epsilon_2<\ldots$. The above definition
implies that the set $\E=\{e_n,\, n\in {\Bbb N}_0\}$ is an
orthonormal basis of $\Hil$, so that it produces a resolution of the
identity which we write in the bra-ket language as
$\sum_{n=0}^\infty |e_n\kt\,\br e_n|=\1$. The following theorem can
be proved:

\begin{thm}
Let  $H$ be well behaved wrt $\Theta$, where $\Theta, \Theta^{-1}\in B(\Hil)$. Then it is possible to introduce two operators $a$ and $b$ on $\Hil$, and a sequence of real numbers $\{\epsilon_n,\,n\in {\Bbb N}_0 \}$, such that the triple $(a,b,\{\epsilon_n\})$ is a family of NLRPB.

Vice versa, if  $(a,b,\{\epsilon_n\})$ is a family of NLRPB, two operators can be introduced,
$H$ and $\Theta$, such that $\Theta, \Theta^{-1}\in B(\Hil)$, and $H$ is well behaved wrt $\Theta$.
\end{thm}

\begin{proof}

To prove the first part of the theorem we introduce the following families of vectors of $\Hil$:
 $$
 \F_\Phi=\left\{\Phi_n:=\Theta^{-1/2}e_n, \,n\in {\Bbb
 N}_0\right\}\,, \quad \F_\eta=\left\{\eta_n:=\Theta^{1/2}e_n, \,n\in
 {\Bbb N}_0\right\}.
 $$
Because of our assumptions, $\F_\Phi$ and $\F_\eta$ are Riesz bases
of $\Hil$ which are also biorthonormal:
$\left<\Phi_n,\eta_k\right>=\delta_{n,k}$. Hence, conditions {\bf
p4} and {\bf p5} are satisfied. On $\F_\Phi$ we define two operators
$a$ and $b$ as follows \be a\Phi_n=\sqrt{\epsilon_n}\,\Phi_{n-1},
\quad b\Phi_n=\sqrt{\epsilon_{n+1}}\,\Phi_{n+1}, \label{71}\en for
all $n\geq0$. By definition, {\bf p1} is satisfied: $a\Phi_0=0$
since $\epsilon_0=0$ and $\Phi_0\in D^\infty(b)$ since, iterating
the second equation in (\ref{71}) we deduce that
$b^n\Phi_0=\sqrt{\epsilon_n!}\Phi_n$, which shows that $b^n\Phi_0$
is well defined for all $n$, being $\Phi_n=\Theta^{-1/2}e_n\in\Hil$.
To check condition {\bf p2} we first have to compute the action of
$a^\dagger$ and $b^\dagger$ on a suitable dense set in $\Hil$. It is
easy to check that (\ref{71}), together with the fact that
$\left<\Phi_n,\eta_k\right>=\delta_{n,k}$, imply that \be
a^\dagger\eta_n=\sqrt{\epsilon_{n+1}}\,\eta_{n+1}, \quad
b^\dagger\eta_n=\sqrt{\epsilon_{n}}\,\eta_{n-1}, \label{72}\en for
all $n\geq0$. Equations (\ref{71}) and (\ref{72}) imply that, for
instance, $b$ is a raising operator for $\F_\Phi$ but $b^\dagger$ is
a lowering operator for $\F_\eta$. From (\ref{72}) we see that
$b^\dagger\eta_0=0$. Iterating the first equation, we also find that
$(a^\dagger)^n\eta_0=\sqrt{\epsilon_n!}\eta_n=\sqrt{\epsilon_n!}\Theta^{1/2}e_n$,
which is a well defined vector in $\Hil$ for all $n\geq0$. Hence
{\bf p2} is satisfied. We end this part of the proof noticing that
{\bf p3'} surely holds and, as a consequence, also {\bf p3} is
verified. Hence $(a,b,\{\epsilon_n\})$ is a family of NLRPB, as
expected.

\vspace{2mm}

Let us now prove the converse: we assume that $(a,b,\{\epsilon_n\})$
is a family of NLRPB and we show how to construct two operators, $H$
and $\Theta$, such that $H$ is cryptohermitian and well behaved wrt
$\Theta$.

This proof  is based on the fact that, since $\F_\Phi$ and $\F_\eta$
are Riesz bases, the operators
 \be
 S_\Phi:=\sum_{n=0}^\infty|\Phi_n\kt\,\br \Phi_n|\,, \ \ \
 S_\eta:=\sum_{n=0}^\infty|\eta_n\kt\,\br \eta_n|
 \label{metroids}
 \ee
are {\em both} positive and bounded. Assuming that
$\left<\Phi_0,\eta_0\right>=1$, they satisfy $S_\Phi=S_\eta^{-1}$.
The operator $H:=ba$ is well (and densely) defined since, because of
{\bf p3}, $\Phi_n\in D(a)$ and $a\Phi_n\in D(b)$. More than this, we
deduce that $H\Phi_n=\epsilon_n\Phi_n$. Analogously we find that
$\eta_n\in D(H^\dagger)$, and that
$H^\dagger\eta_n=\epsilon_n\eta_n$, for all $n\geq0$. Since
$S_\Phi\eta_n=\Phi_n$ and $S_\eta\Phi_n=\eta_n$ we can rewrite this
last eigenvalue equation as $S_\eta^{-1}H^\dagger
S_\eta\,\Phi_n=\epsilon_n\Phi_n$, which, together with the first
eigenvalue equation and using the completeness of $\F_\Phi$, implies
that $H=S_\eta^{-1}H^\dagger S_\eta$. Hence $H$ is CHwrt$S_\eta$.
Due to the properties of intertwining operators $H$, $H^\dagger$ and
$h:=S_\eta^{1/2}\,H\,S_\eta^{-1/2}$ all have the same eigenvalues
and related eigenvectors. This concludes the proof.

\end{proof}

We want to briefly consider few consequences and remarks of this theorem.
\begin{enumerate}

\item The formal expressions of the operators introduced so far can be easily deduced. For instance we have
 \be
a=\sum_{n=0}^\infty \sqrt{\epsilon_n}|\Phi_{n-1}\kt\,\br
\eta_n|,\quad b=\sum_{n=0}^\infty
\sqrt{\epsilon_{n+1}}|\Phi_{n+1}\kt\,\br \eta_n|.
 \label{expr}
 \ee
From these we can also deduce the formal expansions for $a^\dagger$
and $b^\dagger$. Moreover $h=\sum_{n=0}^\infty
\epsilon_n|e_{n}\kt\,\br e_n|$, $H=\sum_{n=0}^\infty
\epsilon_n|\Phi_{n}\kt\,\br \eta_n|$ and
$H^\dagger=\sum_{n=0}^\infty \epsilon_n|\eta_{n}\kt\,\br \Phi_n|$.
These formulas show, among other features, that $h$, $H$ and
$H^\dagger$ are isospectrals.
\item
A straightforward computation shows that $S_\Phi=\Theta^{-1}$ and
$S_\eta=\Theta$, as we have also deduced in the proof of the second
part of the theorem. This fact, together with our previous results,
shows that the frame operators $S_\eta$ and $S_\Phi$, as well as
their square roots, behave as intertwining operators. This is
exactly the same kind of results we can deduce for {\em ordinary}
pseudo-bosons, where biorthogonal Riesz bases and intertwining
operators are recovered.
\item
Even if $h$ is not required to be factorizable, because of our
construction it turns out that it can be written as $h=b_\Theta
a_\Theta$, where $a_\Theta=\Theta^{1/2}a\,\Theta^{-1/2}$ and
$b_\Theta=\Theta^{1/2}b\,\Theta^{-1/2}$. Incidentally, in general
$[a_\Theta,b_\Theta]=\Theta^{1/2}[a,b]\,\Theta^{-1/2}\neq [a,b]$,
but if $\left[[a,b],\Theta^{1/2}\right]=0$, which is the case for
pseudo-bosons. Therefore, at least at a formal level, our
construction shows that the hamiltonian $h$ can be written in a
factorized form.

\item
The reasons for the attention paid to the role of Riesz bases may be
traced back to the Mostafazadeh' results. In chapter 2 of review
\cite{ali} (cf. also references therein, or \cite{SIGMAdva} and
\cite{bagnlpb}) he emphasized that in the methodical analyses of the
formalism of pseudo-hermitian quantum mechanics it makes sense to
pay particular attention to the finite dimensional Hilbert spaces
for simplicity. This inspired not only the present proof but also
the popular constructions of metric operators using Riesz bases
formed by eigenstates of non-Hermitian Hamiltonians. The same idea
also helped to clarify the essence of the problem of the ambiguity
of the metric as formulated by Scholtz et al \cite{Geyer}.

\item
Although we deal here with an infinite-dimensional Hilbert
space in general, it makes good sense to contemplate a reduction of
our observations
to a finite
dimensional Hilbert space.
In such a
simplified scenario one reveals several interesting connections
with the
recent
$n-$level
coherent-state
constructions by Najarbashi et al. \cite{Najar}.

\end{enumerate}

\section{Illustrative matrix models with ascending spectra}

\subsection{A two-by-two example  with two free parameters\label{2by2}}

We consider first a two-dimensional illustrative schematic matrix
example, originally introduced in \cite{bagnlpb}. Let $\Hil={\Bbb
C}^2$ be our Hilbert space and let us consider the following
matrices on $\Hil$,
 \be
A=\left(
   \begin{array}{cc}
     -1 & \beta \\
     -\frac{1}{\beta} & 1 \\
   \end{array}
 \right),\qquad B=\left(
   \begin{array}{cc}
     -1 & \delta \\
     -\frac{1}{\delta} & 1 \\
   \end{array}
 \right),\qquad
 \label{kreace}
 \ee
where $\beta\neq\delta$ are real quantities.  The vectors $\Phi_0=y\left(
           \begin{array}{c}
             \beta \\
             1 \\
           \end{array}
         \right)
$
and $\eta_0=w\left(
           \begin{array}{c}
             1 \\
             -\delta \\
           \end{array}
         \right)
$  satisfy $A\Phi_0=B^\dagger\eta_0=0$ and contain  normalization
constants  $y$ and $w$ which we take real and constrained,
$yw(\beta-\delta)=1$. Putting $\epsilon_0=0$ and
$\epsilon_1=-\frac{1}{\beta\delta}(\beta-\delta)^2$ we  define
 $$
\Phi_1=\frac{1}{\sqrt{\epsilon_1}}\,B\,\Phi_0
       =\frac{y}{\sqrt{\epsilon_1}}\left(
           \begin{array}{c}
             \delta-\beta \\
             -\frac{\beta}{\delta}+1 \\
           \end{array}
         \right), \qquad \eta_1=\frac{1}{\sqrt{\epsilon_1}}\,A^\dagger\,
                   \eta_0=\frac{w}{\sqrt{\epsilon_1}}\left(
           \begin{array}{c}
             \frac{\delta}{\beta}-1 \\
              \beta-\delta\\
           \end{array}
         \right).
 $$
Hence both $\F_\Phi=\{\Phi_0,\Phi_1\}$ and
$\F_\eta=\{\eta_0,\eta_1\}$ are (biorthogonal) Riesz bases of
$\Hil$, satisfying $A\Phi_0=B^\dagger\eta_0=0$,
$A\Phi_1=\sqrt{\epsilon_1}\Phi_0$ and
$B^\dagger\eta_1=\sqrt{\epsilon_1}\,\eta_0$. With this choice,
calling
 \be
M=BA=\left(
       \begin{array}{cc}
         1-\frac{\delta}{\beta} & \delta-\beta \\
         \frac{1}{\delta}-\frac{1}{\beta} & -\frac{\beta}{\delta}+1 \\
       \end{array}
     \right)\qquad \M=A^\dagger B^\dagger=\left(
       \begin{array}{cc}
         1-\frac{\delta}{\beta} & \frac{1}{\delta}-\frac{1}{\beta} \\
         \delta-\beta & -\frac{\beta}{\delta}+1 \\
       \end{array}
     \right),
     \label{calling}
 \ee
we can check that $M\Phi_k=\epsilon_k\Phi_k$ and
$\M\eta_k=\epsilon_k\eta_k$, $k=0,1$. It is also easy to compute
$[A,B]$, which is different from zero if $\delta\neq\beta$ and it is
never equal to the identity operator. Also, we have
$\sum_{k=0}^1|\Phi_k\kt\,\br \eta_k|=\1$ and
 \be
S_\Phi=y^2\left(
       \begin{array}{cc}
         \beta(\beta-\delta) & 0 \\
         0 & 1-\frac{\beta}{\delta} \\
       \end{array}
     \right),\quad
S_\eta=w^2\left(
       \begin{array}{cc}
         1-\frac{\delta}{\beta} & 0 \\
         0 & \delta(\delta-\beta) \\
       \end{array}
     \right).
     \label{bymo}
 \ee
A direct computation finally shows that $S_\Phi=S_\eta^{-1}$ and
that $MS_\Phi =S_\Phi \M$. This can be written as $M=S_\eta^{-1}
\M\,S_\eta$, which shows that $M$ is CHwrt$S_\eta$. Moreover, as it
is clear, $S_\eta$ and $S_\eta^{-1}=S_\Phi$ are bounded operators.
Hence, the first part of the second statement of Theorem 1 is
recovered. To check that $M$ is also well behaved wrt $S_\eta$ it is
sufficient to compute $h=S_\eta^{1/2}MS_\eta^{-1/2}$, and to compute
the two eigenvalues which must have multiplicity 1. This is a simple
exercise in linear algebra and will not be done here.

\subsection{An $N$ by $N$ matrix example without free parameters}

Whenever one tries to apply the principles of cryptohermitian
quantum mechanics in phenomenology, say, of solid-state physics
\cite{Joglekar}, one must contemplates matrices (\ref{57}) of
perceivably larger dimensions $N \gg 2$. In such a realistic setting
one is usually forced to employ a suitable purely numerical method.
Typically, it is practically impossible to employ the
finite-dimensional version
 \be
a=\sum_{n=1}^{N-1} \sqrt{\epsilon_n}|\Phi_{n-1}\kt\,\br
\eta_n|,\quad b=\sum_{n=0}^{N-2}
\sqrt{\epsilon_{n+1}}|\Phi_{n+1}\kt\,\br \eta_n|
 \label{exprN}
 \ee
of the spectral-like expansion formula (\ref{expr}) because its
components themselves are only available, generically, in a purely
numerical representation. The situation further worsens if one tries
to render the ``phenomenological input" matrices (\ref{57}) varying
with a suitable coupling-simulating parameter.

Fortunately, several  arbitrary$-N$ benchmark examples have been
recently found in the context of a cryptohermitian reinterpretation
of certain properties of the classical orthogonal polynomials
\cite{gegenb,lagenre,chebypol}. For our present illustrative
purposes the latter reference proves particularly suitable since it
renders {\em both} the underlying $N$ by $N$ Schr\"{o}dinger
equations, $H\Phi_n=\epsilon_n\Phi_n$ and $H^\dagger\eta_n=\epsilon_n\eta_n$,
exactly solvable.

At the general matrix dimension $N \geq 2$ the main message
delivered by ref.~\cite{chebypol} may be read as the discovery of
the feasibility of the construction of the $N-$parametric metrics
$\Theta$ of which the above-defined matrices $S_\eta$ represent just
the NLPB-related special cases of present interest. In the opposite
direction, the above-mentioned exact solvability of the pair of
Schr\"{o}dinger equations will make it easy, for us, to feel guided
by our previous benchmark example of paragraph \ref{2by2}.


Our present extension of the above illustrative $N=2$ considerations
to all the finite integers $N = 2, 3, \ldots$ will be based on the
results of Ref.~\cite{chebypol} where the Hamiltonian-simulating
matrices were chosen in the form which we shall denote by the tilded
symbol
 \be
 \tilde{H}=\left[
\begin {array}{ccccc}
0&2&0&\ldots&0\\
{}1&0&1&\ddots&\vdots\\
{}0&1&\ddots&\ddots&0\\
{}\vdots&\ddots&\ddots&0&1\\
{}0&\ldots&0&1&0
\end {array} \right]\,.
 \label{Tsqwodel}
 \ee
The standard Hermiticity condition is obviously violated here. In
the notation of Ref.~\cite{chebypol} and via the underlying
conjugate pair of the linear algebraic Schr\"{o}dinger eigenvalue
problems
    \be
     \tilde{H}\Phi_n=E_n\Phi_n,\quad \ \
     \left [\tilde{H}\right ]^\dagger\eta_n
     =E_n\eta_n
     \label{hasdiscus}
     \ee
we get, in particular, the site-indexed components $\{
\alpha|\Phi\kt$, $\alpha=1,2,\ldots,N$ of the $n-$th eigenstate
$\Phi_n$ of our Hamiltonian (where, conventionally, $n = 0, 1,
\ldots, N-1$) in the following closed and arbitrarily normalized
form,
 \ben
 \{ 1|\Phi\kt=T(0,x)= 1\,,\ \ \  \{ 2|\,\Phi\kt=T(1,x)= x
 \,,\ \ \
 \een
 \be
  \{ 3|\Phi\kt=T(2,x)=
 2\,{x}^{2}-1
 \,,\ \ \
 \,,\ldots\,, \{ N|\Phi\kt=T(N-1,x),
 \label{setTa}
 \ee
where the letter $T$ denotes the classical orthogonal Chebyshev
polynomials of the first kind. One can easily deduce \cite{chebypol}
that
$$
\Phi_n=\left(
         \begin{array}{c}
           T(0,x) \\
           T(1,x) \\
           \vdots \\
           \vdots \\
           T(N-1,x) \\
         \end{array}
       \right).
$$
The argument $x=x_n^{(N)}$ of these polynomials is fixed by the
secular equation $T(N,x_n)=0$ which is exactly solvable,
 \be
  E_n=2\,x_n^{(N)}
  =-2\,\cos
 \frac{(n+1/2)\pi}{N}\,, \ \ \ \ \ n = 0, 1, \ldots, N-1\,.
 \label{eigvalT}
 \ee
This formula defines the necessary $N-$plet of energies at every
dimension $N$ (note that in comparison with Ref.~\cite{chebypol} a
more natural and convenient choice of the minus sign is being used
here and in what follows).

For our present purposes we still need to replace the tilded,
auxiliary Hamiltonians $\tilde{H}$ of Eq.~(\ref{Tsqwodel}) (which do
not exhibit the above-required positive-semidefiniteness of the
spectrum) by our following untilded, constant-shifted ultimate
matrices $M$ which are real and manifestly non-Hermitian in the
conventional sense,
 \be
 M=\left[
\begin {array}{ccccc}
{Z} &2&0&\ldots&0\\
{}1&{Z} &1&\ddots&\vdots\\
{}0&1&\ddots&\ddots&0\\
{}\vdots&\ddots&\ddots&{Z} &1\\
{}0&\ldots&0&1&{Z}
\end {array} \right]
 \label{Tsqwodel2}
 \ee
and where we choose ${Z} \equiv E_{N-1}>0$. In terms of these
matrices with the property $M \neq M^\dagger :=M^T$ (the superscript
$^T$ marks transposition) we may write down our final, untilded pair
of toy-model Schr\"{o}dinger equations
    \be
     M\Phi_n=\epsilon_n\Phi_n,\quad \ \
     M^T\eta_n
     =\epsilon_n\eta_n
     \label{hasdiscussed}
     \ee
with the sharply ascending spectrum, $\epsilon_n=E_n+{Z} \geq 0$.

\section{
Closed-form constructions at $N \leq 5$}

One of the most important merits of our toy-model matrix
(\ref{Tsqwodel2}) has been found to lie in the closed form of the
solution of the second, conjugate Schr\"{o}dinger equation in
(\ref{hasdiscussed}). For the corresponding lattice-site
unnormalized components $\{\alpha|\eta\kt$ one obtains the
following, almost identical prescription
 \be
 \{ \alpha|\,\eta\kt=T(n,x)
 \,,\ \ \ \alpha= 2, 3,\ldots, N\,
 \label{setTb}
 \ee
which does not differ from its predecessor (\ref{setTa}) but which
must be complemented by the single different missing item $ \{
1|\eta\kt=T(0,x)/2= 1/2$. The latter feature makes the resulting
biorthogonal system of vectors deceptively similar to an orthogonal
system. Thus, for many purposes it proves useful to separate the
whole set of sites into the ``exceptional" item $\alpha=1$
accompanied by the $(N-1)-$dimensional rest. The more detailed
description of several technical consequences of this split may be
found in Ref.~\cite{chebypol}. A particularly important question of
the appropriate choice of normalization has been analyzed in
Ref.~\cite{SIGMAdva}.

\subsection{The choice of $N=2$}

For our present purposes, it is particularly useful to recall
formula (\ref{exprN}) in its utterly elementary $N=2$ version
 \be
 a= \sqrt{\epsilon_1}|\Phi_{0}\kt\,\br \eta_1|,\quad
  b=
 \sqrt{\epsilon_{1}}|\Phi_{1}\kt\,\br \eta_0|
 \label{expr2}
 \ee
where we have to insert the eigenvalues $\epsilon_0=0$ and
$\epsilon_1=2\sqrt{2}$ of the matrix
 $$
 {M}=
 \left[ \begin {array}{cc} \sqrt {2}&2\\\noalign{\medskip}1&\sqrt
 {2}\end {array} \right]
 $$
and of its transpose. The four properly re-normalized real
eigenvectors of these matrices may be easily found and written, for
typographical reasons, in  transposed form
 \ben
  \Phi_{0}^T=c_{R0}\, [1/\sqrt{2},-1/2]\,,\ \ \
  \Phi_{1}^T=c_{R1}\, [1, 1/\sqrt{2}]\,,\ \ \
  \eta_{0}^T=c_{L0}\, [ 1/\sqrt{2},-1]\,,\ \ \
  \eta_{1}^T=c_{L1}\, [1/2, 1/\sqrt{2}]\,.
 \een
By the Schr\"{o}dinger equations themselves the quadruplet of the
normalization constants $c$ is left arbitrary. In the present NLRPB
context it is assumed that we choose their values is such a way that
our vectors form a biorthonormalized system. At $N=2$ it is then the
matter of elementary algebra to specify, say, $c_{R0} =c_{R1}
=c_{L0} =c_{L1}=1$.

At this point it is necessary to realize \cite{SIGMAdva} that at any
$j = 0,1,\ldots,N-1$ the simultaneous multiplication of $c_{Lj}$ and
division of $c_{Rj}$ by the same constant $\nu_j$ will keep the
biorthonormality and bicompleteness relations unchanged. In this
sense, formulae (\ref{metroids}) define just a very specific,
$\nu_j=1$ metric which is unique. At $N=2$, in particular, our
choice of the normalization leads to the metric
 \be
  S_\eta^{(2)}=|\eta_0\kt\,\br \eta_0|+ |\eta_1\kt\,\br \eta_1|=
  \frac{1}{4}\,
  \left[ \begin {array}{cc} 3&-\sqrt {2}
  \\\noalign{\medskip}-\sqrt {2}&6
  \end {array} \right]\,.
   \label{rometroid}
 \ee
Its eigenvalues $s_\pm=(9\pm \sqrt{17})/8$ are both positive. The
same observation can be also made at the higher dimensions $N$. The
most important conclusion may be already drawn from our first
nontrivial illustration (\ref{rometroid}) which shows that in
contrast to the previous $N=2$ model (\ref{bymo}) of Sec. III.1, the
matrix elements of the {\em generic} metric $S_\eta^{(N)}$ (as well
as of its inverse $S_\Phi^{(2)}$) will be {\em all} non-zero.

Another  choice of the above-mentioned scaling parameters $\nu_j\neq
1$ could be employed converting the metric $S_\eta^{(N)}$ into a
diagonal or sparse matrix. The price to be paid is that the
necessary proof of the positivity of this matrix becomes nontrivial
and dimension-dependent in general. More details (as well as a few
elementary sample constructions of the families of metrics assigned
to our present illustrative zero-parametric Hamiltonians) may be
found in Ref.~\cite{chebypol}.

\subsection{The next special case with $N=3$}

The required insertion in the explicit $N=3$ recipe
 \be
 a=
  \sqrt{\epsilon_1}|\Phi_{0}\kt\,\br \eta_1|+
  \sqrt{\epsilon_2}|\Phi_{1}\kt\,\br \eta_2|
  ,\quad
 b= \sqrt{\epsilon_{1}}|\Phi_{1}\kt\,\br \eta_0|
 +\sqrt{\epsilon_{2}}|\Phi_{2}\kt\,\br \eta_1|
 \label{expr3}
 \ee
may be based on the not too tedious evaluation of the eigenvalues
$\epsilon_0=0$, $\epsilon_1=\sqrt{3}$ and $\epsilon_2=2\sqrt{3}$ and
of the respective eigenvectors
 \ben
  \Phi_{0}^T=c_{R0}\, [1,-\sqrt {3}/2,1/2] \,,\ \ \
  \Phi_{1}^T=c_{R1}\, [-1,0,1]\,,\ \ \
  \Phi_{2}^T=c_{R2}\, [2,\sqrt {3},1]\,,\ \ \
 \een
 \ben
  \eta_{0}^T=c_{L0}\,[1,-\sqrt {3},1]\,,\ \ \
  \eta_{1}^T=c_{L1}\, [1,0,-2]\,,\ \ \
  \eta_{2}^T=c_{L2}\, [1,\sqrt {3},1]\,
 \een
of the matrix
 $$
 {M}=
 \left[ \begin {array}{ccc} \sqrt {3}&2&0\\\noalign{\medskip}1&\sqrt {
3}&1\\\noalign{\medskip}0&1&\sqrt {3}\end {array} \right]
 $$
and of its transpose. One can again proceed in full analogy with the
above $N=2$ example. It is perhaps interesting to add that the
NLRPB-related special metric matrices
 \be
  S_\eta^{(3)}=|\eta_0\kt\br \eta_0|+
   |\eta_1\kt\br \eta_1|+ |\eta_3\kt\br \eta_3|
   \label{trometroid}
 \ee
need not necessarily remain non-sparse. For example, the judicious
choice
 \ben
 c_{R0}=1/3\,,\ \
 c_{R1}=-1/3\,,\ \
 c_{R2}=1/6\,,\ \
 c_{L0}=1\,,\ \
 c_{L1}=1\,,\ \
 c_{L2}=1\,
 \een
of the normalization parameters generates the elementary diagonal
metric (\ref{trometroid}),
 \be
  S_\eta^{(3)}=
  \left[ \begin {array}{ccc}
  3&0&0
  \\\noalign{\medskip}0&6&0
  \\\noalign{\medskip}0&0&6
  \end {array} \right]\,.
   \label{urometroid}
 \ee
The main merit of such a diagonal special case (non-numerically
accessible, in our present model, at any $N$ \cite{chebypol}) may be
seen in the facilitated feasibility of the evaluation of the
manifestly hermitian isospectral Hamiltonian
 \be
 h=h^{(3)}=S_\eta^{1/2}M S_\eta^{-1/2}=
  \left[ \begin {array}{ccc}
  \sqrt{3}&\sqrt{2}&0
  \\\noalign{\medskip}\sqrt{2}&\sqrt{3}&1
  \\\noalign{\medskip}0&1&\sqrt{3}
  \end {array} \right]\,
   \label{urometrd}
 \ee
possessing the following illustrative set of orthonormal
eigenvectors,
 \ben
  e_{0}^T=c_{R0}\, [1,-\sqrt {3}/2,1/2] \,,\ \ \
  e_{1}^T=c_{R1}\, [-1,0,1]\,,\ \ \
  e_{2}^T=c_{R2}\, [2,\sqrt {3},1]\,.\ \ \
 \een
Obviously, all of the alternative normalizations and analogous
insertions in the above-listed formulae remain routine. They will
lead to non-numerical, fully exact formulae. The only remaining
challenge may be seen in the extension of the spectral-like
representation of the operators $a$ and $b$ beyond the ``trivial"
cases, i.e., to $N\geq 4$.

\subsection{The choice of $N=4$}

The $N=4$ version of our main definition ~(\ref{exprN}) requires,
again, the evaluation of the trigonometric-function eigenvalues
$\epsilon_j$ of the matrix
 $$
 {M}=
  \left[ \begin {array}{cccc} \sqrt {2+\sqrt {2}}&2&0&0
\\\noalign{\medskip}1&\sqrt {2+\sqrt {2}}&1&0\\\noalign{\medskip}0&1&
\sqrt {2+\sqrt {2}}&1\\\noalign{\medskip}0&0&1&\sqrt {2+\sqrt {2}}
\end {array} \right]\,.
 $$
It is easy to verify that in terms of the auxiliary constants
$\alpha_0=0$, $\alpha_1=2-\sqrt {2}$, $\alpha_2=\sqrt {2}$ and
$\alpha_3=2$ we can write $\epsilon_j =\alpha_j\sqrt {2+\sqrt {2}}$.
What is less routine is the evaluation of the respective
eigenvectors
 \ben
  \Phi_{0}^T=c_{R0}\, [-\sqrt {2}\sqrt {2+\sqrt {2}},
 \sqrt {2}+1,-\sqrt {2+\sqrt {2}},1]  \,\ \ \
 \een
 \ben
  \Phi_{1}^T=c_{R1}\, [1,-1/2\,\sqrt {2}\sqrt {2+\sqrt {2}}+1/2\,
 \sqrt {2+\sqrt {2}},-1/2\,\sqrt {2},1/2\,\sqrt {2+\sqrt {2}}]\,\ \ \
 \een
 \ben
  \Phi_{2}^T=c_{R2}\, [1,1/2\,\sqrt {2}\sqrt {2+\sqrt {2}}-1/2\,
 \sqrt {2+\sqrt {2}},-1/2\,\sqrt {2},-1/2\,\sqrt {2+\sqrt {2}}]\,\ \ \
 \een
 \ben
  \Phi_{3}^T=c_{R3}\, [\sqrt {2}\sqrt {2+\sqrt {2}},
 \sqrt {2}+1,\sqrt {2+\sqrt {2}},1]\,\ \ \
 \een
 \ben
  \eta_{0}^T=c_{L0}\, [-1/2\,\sqrt {2}\sqrt {2+\sqrt {2}},
 \sqrt {2}+1,-\sqrt {2+\sqrt {2}},1]\,\ \ \
 \een
 \ben
  \eta_{1}^T=c_{L1}\, [1,-\sqrt {2}\sqrt {2+\sqrt {2}}+\sqrt {2+\sqrt {2}},
 -\sqrt {2},\sqrt {2+\sqrt {2}}]\,\ \ \
 \een
 \ben
  \eta_{2}^T=c_{L2}\, [1,\sqrt {2}\sqrt {2+\sqrt {2}}-\sqrt {2+\sqrt {2}},
 -\sqrt {2},-\sqrt {2+\sqrt {2}}] \,\ \ \
 \een
 \ben
  \eta_{3}^T=c_{L3}\, [1/2\,\sqrt {2}\sqrt {2+\sqrt {2}},
 \sqrt {2}+1,\sqrt {2+\sqrt {2}},1]\,.
 \een
Although these formulae remain still extremely elementary, their
generation has been based on the computer-assisted symbolic
manipulations.

\subsection{The last, $N=5$ illustration}

Our next (and also last) Hamiltonian matrix
 $$
 {M}= \left[ \begin {array}{ccccc} \frac{1}{2}\,\sqrt {10+2\,
 \sqrt {5}}&2&0&0&0
 \\\noalign{\medskip}1&\frac{1}{2}\,\sqrt {10+2\,\sqrt {5}}&1&0&0
 \\\noalign{\medskip}0&1&\frac{1}{2}\,\sqrt {10+2\,\sqrt {5}}&1&0
 \\\noalign{\medskip}0&0&1&\frac{1}{2}\,\sqrt {10+2\,\sqrt {5}}&1
 \\\noalign{\medskip}0&0&0&1&\frac{1}{2}\,\sqrt {10+2\,\sqrt {5}}
 \end {array}
 \right]
 $$
forces us to conclude that in spite of the purely non-numerical
character of the recipe (based just on insertions), the $N \geq 5$
explicit formulae  become rather lengthy. The typographical
considerations start to represent, in fact, the main limiting factor
of the presentation of the $N\geq 5$ continuation of the series. For
example, in spite of the existence of closed non-trigonometric
formulae at $N=5$, the pentaplet of energies 0., 0.726542529,
1.902113032, 3.077683536, 3.804226065 is already better represented
numerically. The same comment applies also to the closed-form
eigenvectors, with
 \ben
 \Phi_{0}^T=c_{R0}\, [1,-\frac{1}{4}\,\sqrt {10+2\,\sqrt {5}},\frac{1}{4}+\frac{1}{4}\,
 \sqrt {5},\frac{1}{8}\,\sqrt {10+2\,\sqrt {5}}-\frac{1}{8}\,\sqrt {10+2\,
 \sqrt {5}}\sqrt {5},-\frac{1}{4}+\frac{1}{4}
\,\sqrt {5}]   \,\ \ \
 \een
etc, and with
 \ben
 \eta_{0}^T=c_{L0}\,
 [\frac{1}{2}+\frac{1}{2}\,\sqrt {5},-\frac{1}{4}\,\sqrt {10+2\,\sqrt {5}}-\frac{1}{4}\,\sqrt
{10+2\,\sqrt {5}}\sqrt {5},\frac{3}{2}+\frac{1}{2}\,\sqrt {5},-\frac{1}{2}\,\sqrt {10+2\,
\sqrt {5}},1]
 \een
etc.

\section{Conclusions}

We have shown that two apparently different concepts previously
introduced in the context of quantum mechanics with a non
self-adjoint hamiltonian are strongly related, the one producing the
other under very natural assumptions. We have also analyzed a few
examples to show how the construction works. The analysis of
non-regular NLPB, where unbounded metric operators play a crucial
role, will be considered in the nearest future.

\section*{Acknowledgements}

F.B. acknowledges M.I.U.R. for financial support. M.Z. acknowledges
the support by the GA\v{C}R grant Nr. P203/11/1433, by the M\v{S}MT
``Doppler Institute" project Nr. LC06002 and by the Institutional
Research Plan AV0Z10480505.

\end{document}